\documentclass{article}
\usepackage[a4paper,
            bindingoffset=0.2in,
            left=1in,
            right=1in,
            top=1in,
            bottom=1in,
            footskip=.25in]{geometry}

\usepackage[T1]{fontenc}
\usepackage[utf8]{inputenc}
\usepackage[dvipsnames]{xcolor}
\usepackage{algpseudocode}

\usepackage{listings}
\usepackage{caption}
\usepackage{hyperref}
\usepackage{amsthm}
\usepackage{authblk}

\usepackage{algorithm}
\usepackage{amsmath}
\usepackage{amssymb}
\usepackage{mathtools}
\usepackage{physics}
\usepackage{dsfont}
\usepackage{accents}
\usepackage{xfrac}
\usepackage{xcolor}
\usepackage[ligature, inference]{semantic}
\usepackage{ebproof}
\usepackage{graphicx}
\usepackage{stmaryrd}
\usepackage{xspace}
\usepackage{multirow}
\usepackage{array}
\usepackage{adjustbox}
\usepackage{stackengine}
\usepackage{marvosym}

\usepackage{mdframed}
\mdfsetup{%
linecolor=gray!50,
}

\usepackage{soul}

\usepackage{tikz}
\usepackage{tikzsymbols}
\usetikzlibrary{positioning,fit}
\usetikzlibrary{hobby,backgrounds,calc,trees}
\usetikzlibrary{quantikz2}
\usetikzlibrary{patterns}
\usetikzlibrary{calc}
\usetikzlibrary{shapes.geometric, arrows}

\newcommand{\gatecolor}{gray!30}

\renewcommand\fbox{\fcolorbox{gray!50}{white}}

\definecolor{procedures}{RGB}{202,82,213}

\definecolor{commands}{RGB}{9,10,184}

\newcommand{\ia}{\mathrm{i}}
\newcommand{\ja}{\mathrm{j}}

\newcommand{\xa}{\mathrm{x}}
\newcommand{\sa}{\mathrm{l}}

\newcommand{\bb}{\mathrm{b}}
\newcommand{\ea}{\mathrm{e}}
\newcommand{\da}{\mathrm{d}}

\newcommand{\N}{\mathbb{N}}

\newcommand{\B}{\mathbb{B}}
\newcommand{\Z}{\mathbb{Z}}
\newcommand{\q}{\mathrm{q}}

\newcommand{\qa}{\mathrm{a}}

\newcommand{\bq}{\bar{\q}}
\renewcommand{\op}{\mathrm{op}}
\newcommand{\skp}{\tbt{ skip};}
\newcommand{\tb}[1]{\text{\color{commands}\textnormal{\textbf{#1}}}}
\newcommand{\tbt}[1]{\text{\textnormal{\textbf{#1}}}}
\newcommand{\pn}[1]{\text{\color{procedures}\textnormal{\texttt{#1}}}}
\newcommand{\el}[1]{\textnormal{\small[$#1$]}}
\newcommand{\proc}{\mathrm{proc}}

\newcommand{\decl}{\tbt{decl }}
\newcommand{\call}{\tbt{call }}
\renewcommand{\ket}[1]{\ensuremath{\left| #1\right\rangle}}

\newcommand{\kpsi}{\ket{\psi}}
\newcommand{\kpsip}{\ket{\psi'}}
\newcommand{\kpsipp}{\ket{\psi''}}
\newcommand{\kpsik}{\ket{\psi_k}}
\newcommand{\qcase}[3]{\tbt{qcase } #1 \tbt{ of }\{0\to #2, 1\to #3\}}

\newcommand{\cif}[3]{\tbt{if } #1 \tbt{ then } #2 \tbt{ else }#3}

\newcommand{\ST}{\ensuremath{\mathrm{S}}}

\newcommand{\D}{\mathrm{D}}
\newcommand{\PR}{\ensuremath{\mathrm{P}}}

\newcommand{\U}{\mathrm{U}}
\newcommand{\asg}{{\ \textnormal{\texttt{\textasteriskcentered =}}\ }}

\newcommand{\stdv}{\kpsi, A,f}

\newcommand{\width}{\textnormal{width}}
\newcommand{\fbqp}{\textnormal{\textsc{fbqp}\xspace}}

\newcommand{\bqp}{\textnormal{\textsc{bqp}\xspace}}

\newcommand{\fbqpolylog}{\textnormal{\textsc{fbqpolylog}\xspace}}

\newcommand{\bqpolylog}{\textsc{bqpolylog}\xspace}
\newcommand{\qnc}{\textnormal{\textsc{qnc}}}
\newcommand{\bqnc}{\textnormal{\textsc{fbqnc}}}

\newcommand{\half}{\textnormal{\textsc{half}\xspace}}
\newcommand{\wi}{\textnormal{\textsc{width}\xspace}}

\newcommand{\plp}{\textnormal{\textsc{plp}\xspace}}

\newcommand{\size}[1]{| #1 |}
\newcommand{\len}[1]{\texttt{len}(#1 )}
\newcommand{\lent}[1]{\texttt{len}_{#1 }}
\newcommand{\lentt}[2]{\texttt{len}_{#1 }^{#2}}

\newcommand{\sem}[1]{\llbracket #1 \rrbracket}
\newcommand{\topbot}{\diamond}
\newcommand{\semto}[1]{\stackrel{}{\longrightarrow}}

\newcommand{\uph}{\textnormal{Ph}}
\newcommand{\unot}{\textnormal{NOT}}

\def\orcidID#1{\smash{\href{http://orcid.org/#1}{\protect\raisebox{-1.25pt}{\protect\includegraphics{ORCID_Color.eps}}}}}

\newcounter{program}

\newenvironment{program*}{\begin{equation*}\aligned}{\endaligned\end{equation*}}

\usepackage{thm-restate}

\declaretheorem[name=Theorem]{theorem}
\declaretheorem[name=Lemma]{lemma}
\declaretheorem[name=Definition]{definition}

\declaretheorem[name=Example]{example}

\title{Quantum Programming in Polylogarithmic Time}

\author[1]{Florent Ferrari}
\author[2]{Emmanuel Hainry}
\author[2]{Romain Péchoux}
\author[2]{Mário Silva}
\affil[1]{École Normale Supérieure de Lyon, France}
\affil[2]{Université de Lorraine, CNRS, Inria, LORIA, F-54000 Nancy, France}
\date{
\today\\[2ex]
\small \texttt{florent.ferrari@ens-lyon.fr, \{hainry,pechoux,mmachado\}@loria.fr}}

\begin{document}

\maketitle

\begin{abstract}
Polylogarithmic time delineates a relevant notion of feasibility on several classical computational models such as Boolean circuits or parallel random access machines. As far as the quantum paradigm is concerned, this notion yields the complexity class {\fbqpolylog} of functions approximable in polylogarithmic time with a quantum random access Turing machine.
We introduce a quantum programming language with first-order recursive procedures, which provides the first programming language-based characterization of {\fbqpolylog}.
Each program computes a function in {\fbqpolylog} (soundness) and, conversely, each function of this complexity class is computed by a program (completeness).
We also provide a compilation strategy from programs to uniform families of quantum circuits of polylogarithmic depth and polynomial size, whose set of computed functions is known as $\qnc$, and  recover the well-known separation result ${\fbqpolylog} \subsetneq {\qnc}$.
\end{abstract}

\section{Introduction}
\label{s:intro}

\subsection{Motivation}

Quantum computing is a field of research, which is drawing a great amount of interest as it leverages quantum superposition and interference to obtain computational advantage.
The development of  quantum programming languages is thus a key issue, which raises major technical and conceptual challenges to ensure their physicality.
In order to check that quantum programs can be compiled and executed on a quantum computer, one has to
design restrictions and constraints implying that quantum programs do not break the laws of quantum mechanics, for example, no-cloning of data and unitarity of operators.
In addition, there is a need to tame their complexity in order to ensure their feasibility.
By feasibility, we mean that quantum executions do not use too many ancillary qubits and run in tractable time.

Taking inspiration from the classical world, this kind of issue  has lead to the definition and study of several quantum polynomial time classes.
One of the most striking examples of such classes is {\bqp}, the quantum analog of the class of bounded-error probabilistic polynomial time problems {\textsc{bpp}}.
By Yao's theorem~\cite{Yao1993}, {\bqp} corresponds exactly to what can be computed by uniform families of quantum circuits of polynomial size. This class, as well as its extension to functions, {\fbqp}, have been characterized through various means, including function algebras~\cite{Yamakami20} and first-order programs~\cite{HPS23}.

A natural question is then to study subpolynomial complexity classes.
As for classical programs, it allows to express notions of parallel complexity, which is highly relevant in the quantum setting where program superpositions can be viewed as a kind of parallelism with interferences.
Parallelization can be used to reduce the maximum number of operations performed in total on each qubit. 
This is particularly challenging as qubit fidelity, i.e., the ability of qubits to align with their intended states through time and unitary operations, is a bottleneck of quantum computation~\cite{HuangEtAl19}. While polynomial time algorithms performed in sequence are useful in the \emph{fewer qubits, higher fidelity} setting~\cite{HPS23}, parallelized computation becomes more interesting in the case of \emph{more qubits, lower fidelity}~\cite{LDX19}, as the total number of operations on individual qubits, and their necessary coherence time, scale more slowly.
Moreover, separation results between those small complexity classes could lead to a proof of the quantum advantage, as constant depth quantum circuits have been shown to be strictly more powerful than their classical counterparts~\cite{BGK18}.

In the literature, polylogarithmic (polylog) time has been introduced and studied on Quantum Random-Access Turing Machine (QRATM)~\cite{QRAM2008}. As in the classical case, this definition uses random-access machines, as opposed to standard quantum Turing machines, because of the sub-linearity of time: although the machine cannot read its entire input, it can access any input bit or qubit. On quantum models, polylog time corresponds to problems that a QRATM can solve in a polylog number of steps, leading to the definition of the complexity class {\fbqpolylog}~\cite{yamakami2022,yamakami2024} of functions computable with bounded-error in quantum polylog time.

A main open problem is to design programming languages characterizing such a polylog class abstracting from the low-level considerations (machines, uniformity conditions, etc.)~\cite{Yamakami24-cie}.

\subsection{Contributions}

This paper makes a first step towards solving this problem. Towards that end, we introduce a quantum programming language with first-order recursive procedures, named {\plp} for PolyLog Programs (Figure~\ref{fig:syntax}), on which we obtain the following results:
\begin{itemize}
\item \plp{} programs are terminating (Theorem~\ref{thm:termination}) and reversible (Theorem~\ref{thm:reversibility}).
\item \plp{} is sound for {\fbqpolylog} (Theorem~\ref{thm:soundnesscompleteness}), i.e., each \plp{} program computes a function in {\fbqpolylog}. Soundness relies on the use of a bounded recursion scheme for procedures to enforce the required polylog time properties,
as illustrated by the binary search and square-log examples of Figure~\ref{fig:lfoq-examples}.
\item \plp{} is  complete for {\fbqpolylog} (Theorem~\ref{thm:soundnesscompleteness}), i.e., each function of this complexity class is computed by a \plp{} program. Completeness is shown by a direct encoding of polylog QRATM in $\plp$.
\item  $\plp$ is also sound but not complete for $\qnc$. For soundness, we outline a compilation algorithm that from a {\plp} program and its input size outputs a quantum circuit of polylog depth and polynomial size, i.e., circuits computing functions in {\qnc}  (Theorem~\ref{thm:comp}). 
There is an implementation of this algorithm available at \url{https://gitlab.inria.fr/mmachado/pfoq-compiler}, whose compiler works for programs even beyond the $\plp$ fragment, namely non-polylog programs.
Completeness does not hold (Theorem~\ref{thm:more-than-polylog}) as it is well-known that $\fbqpolylog$ is strictly included in $\qnc$.
\end{itemize}

\subsection{Related Work}

Different characterizations of quantum complexity classes have been obtained for polynomial time using, non-exhaustively, lambda-calculus~\cite{dLMZ10}, function algebra~\cite{Yamakami20}, and a first-order programming language~\cite{HPS23}.
A characterization of \bqpolylog{} based on a function algebra has been provided in~\cite{yamakami2022,yamakami2024}, where a \emph{fast quantum recursion scheme} is used to ensure that programs terminate in polylogarithmic time. Our work employs a similar bounded recursion scheme, using a simple divide-and-conquer strategy on qubits. This characterization can be seen as simpler and more natural approach since an imperative first-order programming language is more accessible to the typical programmer.
Similar divide-and-conquer strategies have been explored in quantum computation not only for the purpose of finding quantum advantage~\cite{BGK18,childs2022} but also to leverage classically-inspired techniques in the quantum scenario, where reversibility and unitarity must be satisfied~\cite{VBE96, TK08}. Our work is mostly focused on this second aspect, where the objective is to balance the expressivity of a quantum programming language with the statically-verifiable properties of its programs, namely unitarity, time complexity, as well as circuit size and depth.

\section{First-Order Quantum Programs}
\label{s:lfoq}


\begin{figure}
\begin{center}
\fbox{
\begin{tabular}{llcl}
(Integers) & $\ia$ & $\triangleq$ & $\xa \mid k \mid \ia \pm k \mid \ia/2 \mid \size{\sa} $\\
(Booleans) & $\bb$ &$\triangleq$&$\ia \geq \ia\mid \ldots \mid \bb \wedge \bb \mid \ldots$\\
(Qubit lists) & $\sa$ &$\triangleq$  & $\bar{\q} \mid  \sa\ominus \el{\ia} \mid  \sa^\boxminus \mid \sa^{\boxplus} $\\
(Qubits) & $\q$ & $\triangleq$ & $\sa\el{\ia} $\\
(Statements) & $\ST$ &$\triangleq$ & $\tbt{skip}; \mid  \q \asg \U^g(\ia); \mid \ST\ \ST \mid \tbt{if }\bb\tbt{ then }\{\ST\}\tbt{ else }\{\ST \}$\\
& & & $\mid  \qcase{\q}{\ST}{\ST} $\\
& & &$ \mid \call \proc(\sa_1,\ldots,\sa_{n});$\\
(Procedure decl.) & $\D$ &$\triangleq$ & $\varepsilon \mid  \decl \proc(\bar{\q}_1, \ldots,\bar{\q}_n)\{\ST\},\,\D$ \\
(Programs) & $\PR
$ &$\triangleq$ &$\D::\ST $
\end{tabular}}
\end{center}
\caption{Syntax of programs}
\label{fig:syntax}
\end{figure}

\subsection{Syntax}
The considered language is a quantum programming language with first-order recursive procedures whose syntax is provided in Figure~\ref{fig:syntax}. There are four basic types $\tau$ for expressions: 
\begin{enumerate}
\item \emph{Integer} expressions are variables $\xa$, constant $k \in \mathbb{N}$, arithmetic operations like $\ia \pm k$ or $\ia / 2$,\footnote{The semantics of $/2$ will be defined as the ceiling of the result, hence it preserves the set of integers.} as well as the size  $\size{\sa}$ of a list of qubits $\sa$.
\item \emph{Boolean} expressions are defined in a standard way.
\item \emph{Qubit lists} are lists of unique (i.e., non-duplicable) qubit pointers. 
A qubit list expression $\sa$ is either a variable $\bar{\q}$, the first (respectively second) half $\sa^\boxminus$ (resp. $\sa^\boxplus$) of the qubit list $\sa$, or a list $\sa\ominus\el{\ia}$ where the $\ia$-th element of $\sa$ has been removed. We will also use some syntactic sugar for removing multiple elements of a list with $\sa\ominus\el{\ia_1, \ldots, \ia_k}$.
\item \emph{Qubit} expressions are of the shape $\sa\el{\ia}$, which denotes the $\ia$-th qubit in $\sa$. We also define syntactic sugar for pointing to the $n$-th \emph{last} qubit in a list, by defining for any $n\geq1$, $\bar{\q}\el{-n}\triangleq \bar{\q}\el{|\bar{\q}|-n+1}$.
\end{enumerate}
Throughout the paper, $\ea$, $\da$, $\ldots$ will denote arbitrary expressions of any type.  
Given a syntactic object $t$, let $Var(t)$ be the set of qubit  variables used in $t$, e.g., $Var(\bar{\q}\ominus[2,3])=\{\bar{\q}\}$ is the set of qubit variables in the expression $\bar{\q}\ominus[2,3]$.

A program $\PR
\triangleq \D :: \ST$ is defined in Figure~\ref{fig:syntax} as a list of (possibly recursive) procedure declarations $\D$, followed by a program statement $\ST$. We assume that $Var(\ST) \subseteq Var(\PR)$ holds. 
In what follows, it will be convenient to order the set $Var(\PR)=\{\bar{\q}_1 , \ldots ,\bar{\q}_m \}$ by $\bar{\q}_1 < \ldots <\bar{\q}_m$ to fix a precise memory representation of quantum states.

Let Procedures be an enumerable set of procedure names $\proc$. We write $\proc\in\PR$ to denote that $\proc$ appears in $\PR$.  
Each procedure of name $\proc \in \PR$ is defined by a unique procedure declaration $\tbt{decl }\proc(\bar{\q}_1, \ldots,\bar{\q}_n)\{\ST\}\in \D$, which takes the lists of qubits $\bar{\q}_i$ as parameters.
It always holds that $Var(\ST) \subseteq \{\bar{\q}_1, \ldots,\bar{\q}_n\}$. We will sometimes write $\ST^{\proc}$ to explicitly state that $\ST$ is the statement of $\proc$.

Statements include the no-op instruction, (single-qubit) unitary application, sequences, conditional, quantum case, and procedure calls. For sake of universality~\cite{Bokyn1999}, in a unitary application $\q \asg \U^g(\ia);$, the unitary transformation $\U^g(\ia)\in\mathbb{K}^{2\times 2}$ can take an integer $\ia$ and a function $g \in \mathbb{Z} \to [0,2\pi)$ as optional arguments\footnote{In the case of quantum polynomial time, Adleman et al.~\cite{ADH97} showed how the choice of amplitudes can affect the expressivity of classes such as \bqp{}, requiring the restriction of \emph{polynomial-time approximable complex amplitudes}. How the set of amplitudes influences the class \fbqpolylog{} remains an open question, as discussed in~\cite{yamakami2022,yamakami2024}, and so we abstain from the use of the entire set of complex numbers and instead use a field $\mathbb{K}$ which may refer to, for instance, polynomial-time approximable complex amplitudes.}, and we omit them when they are of no use. As we shall see in next section, unitary transformations will be restricted to phase gate, rotation gates, and NOT gate.

 The quantum conditional $\tbt{qcase }\q \tbt{ of }\{0\to \ST_0,1\to\ST_1\}$ allows branching by executing statements $\ST_0$ and $\ST_1$ in superposition according to the state of qubit $\q$.
When we want to treat cases on multiple qubits, we will sometimes simplify the nested qcases, for example $\tbt{qcase }\q_1, \q_2 \tbt{ of }\{00 \to \ST_{00}, 01 \to \ST_{01}, 10 \to \ST_{10}, 11 \to\ST_{11}\}$ is a shorthand notation for
\newpage
\begin{lstlisting}[numbers=none,keywordstyle=\rmfamily\bfseries]
qcase $\q_1$ of {
    0 $\to$ qcase $\q_2$ of {
        0 $\to$ $\ST_{00}$,
        1 $\to$ $\ST_{01}$
    },
    1 $\to$ qcase $\q_2$ of {
        0 $\to$ $\ST_{10}$,
        1 $\to$ $\ST_{11}$
    }
}
\end{lstlisting}

In each procedure call $\call \proc (\sa_1,\ldots,\sa_n);$, the no-cloning theorem of quantum mechanics imposes the restriction that $\forall i \neq j,\ Var(\sa_i) \neq Var(\sa_j)$.

The syntax of Figure~\ref{fig:syntax} can be used to define typical quantum computing primitives such as controlled-NOT, swap, as well as Toffoli gates, as syntactic sugar:
\begin{align*}
\textrm{CNOT}(\q_1,\q_2)&\triangleq\  \qcase{\q_1}{\skp}{\q_2\asg \textrm{NOT};}\\
\textrm{SWAP}(\q_1,\q_2)&\triangleq\  \textrm{CNOT}(\q_1,\q_2)\ \textrm{CNOT}(\q_2,\q_1)\ \textrm{CNOT}(\q_1,\q_2)\\
\textrm{TOF}(\q_1,\q_2,\q_3)&\triangleq\ \qcase{\q_1}{\skp}{\textrm{CNOT}(\q_2,\q_3)}
\end{align*}

\subsection{Semantics}

Let $\mathbb{B}\triangleq \{0,1\}$ denote the set of Booleans and $\mathcal{L}(\mathbb{N})$ denote the set of lists of natural numbers, $[]$ being the empty list. We interpret basic types as follows:
\begin{align*}
\sem{\text{Integers}} &\triangleq \mathbb{Z} &
\sem{\text{Booleans}} &\triangleq \mathbb{B}  &\sem{\text{Qubit lists}} &\triangleq \mathcal{L}(\mathbb{N})   &\sem{\text{Qubits}} &\triangleq \mathbb{N} 
\end{align*}
Qubits are interpreted as integers (pointers) and qubit lists are interpreted as lists of pointers.
Each $\op \in \{\pm, /2, \geq,\wedge,\ldots\}$ of arity $n$ comes with a basic type signature $\op:: \tau_1 \times \ldots \times \tau_n \to \tau$ and computes a fixed total function $\sem{\op} \in \sem{\tau_1} \times \ldots \times \sem{\tau_n} \to \sem{\tau}$. We set $\op(\tau_1 ,\ldots, \tau_n) \triangleq \tau$.  For example, $\sem{/2} \triangleq m \mapsto \lceil \sfrac{m}{2}\rceil$. Constants $k$ are treated as particular operators of arity $0$.
Given a program $\PR$, for each basic type $\tau$, the semantics of expressions is described standardly in Figure~\ref{table:semnatbool} as a function 
\[\Downarrow_{\sem{\tau}}\ :  \tau \times (Var(\PR) \to \mathcal{L}(\mathbb{N})) \to \sem{\tau}.\]
$(\ea,f)\Downarrow_{\sem{\tau}} v$ holds when expression $\ea$ of type $\tau$ evaluates to the value $v\in\sem{\tau}$ under the context $f \in Var(\PR) \to \mathcal{L}(\mathbb{N})$. The context $f$ is just a map from each program input to a list of qubit pointers taken into consideration when evaluating $\ea$. For instance, we have that $(\bar{\q}\el{2},\bar{\q} \mapsto [1,4,5])\Downarrow_\mathbb{N} 4$ (the second qubit is of index $4$), $(\bar{\q}\ominus \el{4},\bar{\q} \mapsto[1,4,5])\Downarrow_{\mathcal{L}(\mathbb{N})} []$ ($[]$ is used for errors on type $\mathcal{L}(\mathbb{N})$), $(\bar{\q}\el{4},\bar{\q} \mapsto[1,4,5])\Downarrow_\mathbb{N} 0$ (index $0$ is used for error on type $\mathbb{N}$), and $(\bar{\q}\ominus\el{3},\bar{\q} \mapsto[1,4,5])\Downarrow_{\mathcal{L}(\mathbb{N})} [1,4]$ (the third qubit has been removed).

\begin{figure}
\begin{mdframed}
\centering
$
\begin{prooftree}
\hypo{\forall i\leq n,\ (\ea_i, f) \Downarrow_{\sem{\tau_i}} x_i}
\infer1[
]{(\op(\ea_1,\ldots,\ea_n), f) \Downarrow_{\sem{\op(\tau_1,\ldots,\tau_n)}} \sem{\op}(x_1,\ldots,x_n)}
\end{prooftree}
\qquad 
\begin{prooftree}
\hypo{(\sa,f)\Downarrow_{\mathcal{L}(\mathbb{N})} [x_1,\ldots,x_n]}
\infer1[
]{(\size{\sa},f)\Downarrow_{\Z} n}
\end{prooftree}
$
\\[4mm]
$
\begin{prooftree}
\hypo{\bar{\q}\in Var(\PR)}
\infer1[
]{(\bar{\q}, f) \Downarrow_{\mathcal{L}(\mathbb{N})} f(\bar{\q})}
\end{prooftree}
\qquad
\begin{prooftree}
\hypo{(\sa, f)\Downarrow_{\mathcal{L}(\mathbb{N})} [x_1, \ldots, x_m]}
\hypo{(\ia, f) \Downarrow_{\Z} k \in \{1,\ldots,m\}}
\infer2[
]{(\sa\ominus\el{\ia}, f) \Downarrow_{\mathcal{L}(\mathbb{N})} [x_1, \ldots, x_{k-1}, x_{k+1}, \ldots, x_m]}
\end{prooftree}
$
\\[4mm]
$
\begin{prooftree}
\hypo{(\sa, f)\Downarrow_{\mathcal{L}(\mathbb{N})} [x_1, \ldots, x_m]}
\hypo{(\ia, f) \Downarrow_{\Z} k \notin \{1,\ldots,m\}}
\infer2[
]{(\sa\ominus\el{\ia}, f) \Downarrow_{\mathcal{L}(\mathbb{N})} [\,]}
\end{prooftree}
\qquad
\begin{prooftree}
\hypo{(\sa, f)\Downarrow_{\mathcal{L}(\mathbb{N})} [x_1, \ldots, x_m]}
\hypo{m>1}
\infer2[
]{(\sa^\boxminus, f) \Downarrow_{\mathcal{L}(\mathbb{N})} [x_1,\dots,x_{\lceil\sfrac{m}{2}\rceil}]}
\end{prooftree}
$
\\[4mm]
$
\begin{prooftree}
\hypo{(\sa, f)\Downarrow_{\mathcal{L}(\mathbb{N})} [x_1, \ldots, x_m]}
\hypo{m>1}
\infer2[
]{(\sa^\boxplus, f) \Downarrow_{\mathcal{L}(\mathbb{N})} [x_{\lceil\sfrac{m}{2}\rceil+1},\dots,x_m]}
\end{prooftree}
\qquad
\begin{prooftree}
\hypo{(\sa, f)\Downarrow_{\mathcal{L}(\mathbb{N})} [x_1, \ldots, x_m]}
\hypo{(\ia, f) \Downarrow_{\Z} k \notin \{1,\ldots,m\}}
\infer2[
]{(\sa\el{\ia}, f) \Downarrow_{\mathbb{N}} 0}
\end{prooftree}
$
\\[4mm]
$
\begin{prooftree}
\hypo{(\sa, f)\Downarrow_{\mathcal{L}(\mathbb{N})} l}
\hypo{\size{l} \leq 1}
\hypo{\pm\in\{{\scriptstyle\boxminus, \boxplus}\}}
\infer3[
]{(\sa^\pm, f) \Downarrow_{\mathcal{L}(\mathbb{N})} [\,]}
\end{prooftree}
\ \quad
\begin{prooftree}
\hypo{(\sa, f)\Downarrow_{\mathcal{L}(\mathbb{N})} [x_1, \ldots, x_m]}
\hypo{(\ia, f) \Downarrow_{\Z} k \in \{1,\ldots,m\}}
\infer2[
]{(\sa\el{\ia}, f) \Downarrow_{\mathbb{N}} x_k}
\end{prooftree}
$
\end{mdframed}
\caption{Semantics of expressions}
\label{table:semnatbool}
\end{figure}

Let $\mathcal{H}_{2^n} $ denote the Hilbert space  $\mathbb{C}^{2^n}$ of $n$ qubits with tensor product $\otimes$ and let $\mathcal{P}(\mathbb{N})$ denote the powerset of $\mathbb{N}$.  
Given a program $\PR$, let the \emph{length} of $\PR$ be a function mapping each qubit variable  $\bar{\q} \in Var(\PR)$ to an integer  $\len{\bar{\q}} \in \mathbb{N}$. We write $\lent{\PR}$ as a shorthand for $\sum_{\bar{\q} \in Var(\PR)}\len{\bar{\q}}$ and $\lentt{\PR}{<\bar{\q}}$ as a shorthand for $\sum_{\bar{\q}' \in Var(\PR),\ \bar{\q}'<\bar{\q}}\len{\bar{\q}'}$.

A \emph{configuration} $c$ of program $\PR$ over $\lent{\PR}$ qubits is of the shape
\[(\ST,\stdv) \in (\text{Statements}\cup \{\top,\bot\}) \times \mathcal{H}_{2^{\lent{\PR}}} \times (Var(\PR) \to \mathcal{P}(\mathbb{N})) \times (Var(\PR) \to \mathcal{L}(\mathbb{N})),\]
where $\top$ and $\bot$ are two special symbols denoting termination and error, respectively, where $\kpsi \in \mathcal{H}_{2^{\lent{\PR}}}$ is a quantum state, and where, for each qubit list $\bar{\q} \in Var(\PR)$, $A(\bar{\q})$ is the set of qubit pointers accessible  from  $\bar{\q}$ and $f(\bar{\q})$ is the list of qubit pointers assigned to $\bar{\q}$.
Given a qubit $\q$ such that $Var(\q)=\{\bar{\q}\}$, we write $A(\q)$ as a shorthand for $A(\bar{\q})$ and we write $A_{\q \backslash \{n\}}$ for the function $A' \in Var(\PR) \to \mathcal{P}(\mathbb{N})$ defined by $A'(\bar{\q}') \triangleq A(\bar{\q}'),$ $\forall \bar{\q}' \neq \bar{\q}$, and $A'(\bar{\q})\triangleq A(\bar{\q})\backslash \{n\}$.

Given a program $\PR \triangleq \D::\ST$, with $n=\lent{\PR}$, let $\text{Conf}_{n}$ be the set of configurations over $n$ qubits. The initial configuration in $\text{Conf}_{n}$ on input state $\kpsi \in \mathcal{H}_{2^n}$ is $c_{init}(\kpsi) \triangleq (\ST,\kpsi,\bar{\q} \mapsto \{1,\dots,\len{\bar{\q}}\},\bar{\q} \mapsto [1,\dots,\len{\bar{\q}}]) $. A final configuration can be defined in the same way as 
$c_{final}(\kpsi) \triangleq (\top,\kpsi,\bar{\q} \mapsto \{1,\dots,\len{\bar{\q}}\},\bar{\q} \mapsto [1,\dots,\len{\bar{\q}}])$.

Each unitary transformation $\U$ of a unitary application $\q \asg \U^g(\ia);$ comes with a function $\sem{\U}$ assigning a unitary matrix $\sem{\U}(g)(n) \in \mathbb{K}^{2 \times 2}$ to each integer $n$ and function $g \in \mathbb{Z} \to [0,2\pi)$. We restrict ourselves to three kinds of gates: the phase gate $Ph$, rotation gate $\mathrm{R_Y}$ and $NOT$ gate $\mathrm{NOT}$, whose semantics is defined as follows:
\begin{align*}
\sem{\textnormal{Ph}}(g)(n)& \triangleq \begin{pmatrix} 1 &0 \\ 0 & e^{ig(n)} \end{pmatrix}\\[0.2cm] 
\sem{\mathrm{R_Y}}(g)(n) &\triangleq \begin{pmatrix} \cos(g(n)) &-\sin(g(n)) \\ \sin(g(n)) & \cos(g(n)) \end{pmatrix}\\[0.2cm] 
\sem{\mathrm{NOT}}(\cdot)(\cdot) &\triangleq \begin{pmatrix} 0 &1 \\ 1 & 0 \end{pmatrix}
\end{align*}

\begin{figure}[t]
\begin{mdframed}
\centering
$
 \begin{prooftree}
 \hypo{\vphantom{(\sa_k\el{\ia},l)\Downarrow_{\mathbb{N}} n \notin A_k}}
 \infer1[
 ]{(\tbt{skip};,\stdv)\semto{0} (\top,\stdv)}
 \end{prooftree}
\qquad
\begin{prooftree}
 \hypo{(\q,f)\Downarrow_{\mathbb{N}} n \notin A(\q)}
 \infer1[
 ]{(\q\asg \U^g(\ja);\!,\stdv)\semto{0} (\bot,\stdv)}
 \end{prooftree}
$
\\[4mm]
$
\begin{prooftree}
\hypo{(\q,f)\Downarrow_{\mathbb{N}} n \in A(\q)}
\hypo{(\ia,f)\Downarrow_{\mathbb{N}} m}
\hypo{j=\lentt{\PR}{<\q}+n}
 \infer3[
 ]{(\q \asg \U^g(\ia);\!,\stdv)\semto{0} (  \top, I_{2^{j-1}}\otimes \sem{\U}(g)(m)\otimes I_{2^{\lent{\PR}-j}} \kpsi ,A, f )}
 \end{prooftree}
$
\\[4mm]
$
 \begin{prooftree}
  \hypo{(\ST_0,\stdv)\semto{m_1} (\top,\kpsip,A,f)}
   \hypo{(\ST_1,\kpsip,A,l)\semto{m_2} (\topbot,\kpsipp,A,f)}
   \hypo{\topbot \in \{\top,\bot\}}
 \infer3[
 ]{(\ST_0\ \ST_1,\stdv)\semto{m_1+m_2} ( \topbot,\kpsipp,A,f)}
 \end{prooftree}
$
\\[4mm]
$
\begin{prooftree}
  \hypo{( \ST_0,\stdv)\semto{m} (\bot,\stdv)}
  \infer1[
  ]{( \ST_0\ \ST_1,\stdv)\semto{m}(\bot,\stdv)}
 \end{prooftree}
$
\\[4mm]
$
\begin{prooftree}
  \hypo{(\bb,f) \Downarrow_{\B} b}
  \hypo{(\ST_b,\stdv)\semto{m_b} (\topbot,\kpsip,A,f)}
     \hypo{\topbot \in \{\top,\bot\}}
  \infer3[
  ]{(\cif{\bb}{\{\ST_{1}\}}{\{\ST_{0}\}} ,\stdv)\semto{m_b}(\topbot ,\kpsip,A,f)}
  \end{prooftree}
$
\\[4mm]
$
\scalebox{0.975}{
\begin{prooftree}
  \hypo{(\q,f) \Downarrow_{\mathbb{N}} n \in A(\q)}
  \hypo{\forall k \in \mathbb{B},\ (\ST_{k},\kpsi,A_{\q \backslash \{n\}},f)\semto{m_k} (\top,\kpsik,A_{\q \backslash \{n\}},f)}
  \hypo{j=\lentt{\PR}{<\q}+n}
  \infer3[
  ]{ ( \qcase{\q}{\ST_0}{\ST_1},\stdv)\semto{\max_{k\in \{0,1\}} m_k}(\top,\sum_{k\in \{0,1\}} {\ket{k}}_{j} \bra{k}_{j} \kpsik,A,f)}
\end{prooftree}
}
$
\\[4mm]
$
\begin{prooftree}
  \hypo{(\q,f) \Downarrow_{\mathbb{N}} n \in A(\q)}
  \hypo{\exists k \in \mathbb{B},\ (\ST_{k},\kpsi,A_{\q \backslash \{n\}},f)\semto{m_k} (\bot,\kpsik,A_{\q \backslash \{n\}},f)}
  \infer2[
  ]{( \qcase{\q}{\ST_0}{\ST_1},\stdv)\semto{\max_{k\in \{0,1\}} m_k}(\bot,\stdv)}
  \end{prooftree}
$
\\[4mm]
$
\begin{prooftree}
  \hypo{(\q,f) \Downarrow_{\mathbb{N}} n \notin A(\q)}
  \infer1[
  ]{( \qcase{\q}{\ST_0}{\ST_1},\stdv)\semto{0}(\bot,\stdv)}
  \end{prooftree}
$
\\[4mm]
$
\begin{prooftree}
  \hypo{\forall j \leq n,\ (\sa_j,f) \Downarrow_{\mathcal{L}(\mathbb{N})} l_j \neq [\,] }
  \hypo{(\ST^{\proc}\{\sa_j/\bar{\q}_j\},\kpsi,A, f )\semto{m} (\topbot,\kpsip,A,f)}
  \hypo{\topbot \in \{\top,\bot\}}
  \infer3[
  ]{(\call \proc (\sa_1,\ldots,\sa_n);,\stdv)\semto{m+1} (\topbot,\kpsip,A,f)}
\end{prooftree}
$
\\[4mm]
$
\begin{prooftree}
  \hypo{\exists j \leq n,\ (\sa_j,f) \Downarrow_{\mathcal{L}(\mathbb{N})} [\,]}
  \infer1[
  ]{(\call \proc (\sa_1,\ldots,\sa_n);,\stdv)\semto{m+1} (\top,\kpsi,A,f)}
\end{prooftree}
$
\end{mdframed}
\caption{Semantics of statements}
\label{table:operationalsemantics}
\end{figure}

The big-step semantics $\cdot \semto{ } \cdot $ is defined in Figure~\ref{table:operationalsemantics} as a relation in $\bigcup_{n \in \mathbb{N}} \text{Conf}_n  \times \text{Conf}_n$.
The symbols $\bot$ and $\top$ for error and termination, respectively, are terminal states: they cannot appear on the left-hand-side of a rule. 
In Figure~\ref{table:operationalsemantics}, the function $A$ of accessible qubits is used to ensure that unitary operations on qubits can be physically implemented. For example, the statements $\ST_0$ and $\ST_1$ of a quantum branch $\tbt{qcase }{\q}\tbt{ of }\{0\to\ST_0,\,1\to \ST_1\}$ cannot access the control qubit $\q$ to ensure reversibility.
To avoid cumbersome use of nested conditionals in the program syntax, each procedure call on a (at least one) empty qubit list is semantically equivalent to a$\skp$ (see Figure~\ref{table:operationalsemantics}).

We write $\sem{\PR}(\kpsi) = \kpsip$, whenever $c_{init}(\kpsi) \semto{m}c_{final}(\kpsip)$ holds. If the program $\PR$ terminates on all inputs (i.e., always reaches a final configuration) then $\sem{\PR}$ is a total function on quantum states.
Note that if a program terminates then it is obviously error-free (i.e., does not reach a configuration with a $\bot$) but the converse property does not hold. However, every program $\PR$ can be efficiently transformed into an error-free program $\PR_{\neg\bot}$ such that $\forall \kpsi$, if $\sem{\PR}(\kpsi)$ is defined then $\sem{\PR}(\kpsi)=\sem{\PR_{\neg\bot}}(\kpsi)$. This can be done, for instance, by checking the size of qubit lists before accessing them. Hence we will restrict ourselves to error-free programs in what follows.

When an input state is defined by different qubit lists, we denote them in subscript. For instance, for $x,y\in\{0,1\}^\ast$ and $m\in\mathbb{N}$, we have that $\ket{x}_{\bar{\q}_1}\ket{y}_{\bar{\q}_2}$ indicates state $\ket{x}$ given as input to qubit list $\bar{\q}_1$, state $\ket{y}$ given as input to qubit list $\bar{\q}_2$.

\begin{figure}
\begin{mdframed}[innerbottommargin=0pt,skipbelow=-5mm]
\begin{center}
\begin{minipage}[t]{0.45\textwidth}
\underline{\texttt{SEARCH}}
\begin{lstlisting}
$\tb{decl }\pn{search}(\bar{\q}_1,\bar{\q}_2)\,\{$
  $\tb{if }|\bar{\q}_1|>1 \tb{ then }\{$
    $\tb{qcase }\bar{\q}_1\el{|\bar{\q}_1|/2,|\bar{\q}_1|/2+1}\tb{ of }\{$
      $00 \to \tb{call }\pn{search}(\bar{\q}_1^\boxplus\ominus\el{1},\bar{\q}_2);,$
      $01 \to \bar{\q}_2\el{1}\asg \unot;,$
      $10 \to \tb{call }\pn{search}(\bar{\q}_1^\boxminus\ominus \el{-1},\bar{\q}_2);,$
      $11 \to \tb{skip};\,\}$
  $\}\tb{ else }\{\tb{ skip};\}$
$\}\,::$
$\tb{call }\pn{search}(\bar{\q}_1, \bar{\q}_2);$
\end{lstlisting}
\end{minipage}
\begin{minipage}{0.05\textwidth}
\ 
\end{minipage}
\begin{minipage}[t]{0.3\textwidth}
\underline{\texttt{SQLOG}}
\begin{lstlisting}
$\tb{decl }\pn{f}(\bq_1, \bq_2)\,\{$
  $\bq_1\el{1}\asg\uph^{\lambda x.2\pi/x}(\size{\bq_1});$
  $\tb{call }\pn{f}(\bq_1^\boxplus , \bq_2);$
  $\tb{call }\pn{g}(\bq_1, \bq_2\ominus\el{1});\,\}$
$\tb{decl }\pn{g}(\bq_1, \bq_2)\,\{$
  $\tb{qcase }\bq_1\el{\size{\bq_1}/2}\tb{ of }\{$
    $0 \to \tb{call }\pn{g}(\bq_1^\boxminus , \bq_2);,$
    $1 \to \bq_2[1]\asg \unot;$
$\}\}\,::$
$\tb{call }\pn{f}(\bq_1, \bq_2);$
\end{lstlisting}
\end{minipage}
\end{center}
\end{mdframed}
\caption{Examples of {\plp} programs}
\label{fig:lfoq-examples}
\end{figure}

\begin{example}[Binary search]\label{ex:binary-search}
Let $x \in 0^*1^*2^*$ be a sorted string and $\hat{x}$ denote the encoding of $x$ as a binary given by \(\hat{0}\triangleq 00\), \(\hat{1}\triangleq 01\), and \(\hat{2}\triangleq 10\). Program \texttt{SEARCH} in Figure~\ref{fig:lfoq-examples} computes the function $\llbracket\texttt{SEARCH}\rrbracket(\ket{\hat{x}}_{\bar{\q}_1}\ket{0}_{\bar{\q}_2})= \ket{\hat{x}}_{\bar{\q}_1}\ket{b}_{\bar{\q}_2}$, where $b\in\{0,1\}$ indicates whether $x$ contains a $1$ or not.
\end{example}

\begin{example}[Square Log]\label{ex:square-log}
Program \texttt{SQLOG} of Figure~\ref{fig:lfoq-examples} defines two recursive procedures \pn{f} and \pn{g}
and its complexity is square logarithmic in the size of the first qubit list $\bq_1$.
The procedure \pn{g} has logarithmic complexity in $\size{\bq_1}$ as each recursive call to \pn{g} divides the size of the first argument by 2.
Similarly, the procedure \pn{f} calls itself a logarithmic number of times and calls \pn{g} each time, hence accounting for a $O(\log^2(\size{\bq_1}))$ complexity.
\end{example}

\subsection{Polylogarithmic Time Restrictions}
We now define some restrictions on the admissible programs to guarantee that they terminate in polylogarithmic time (i.e., each procedure cannot perform more than a polylogarithmic number of recursive calls in the input size) and that their total number of sequential procedure calls (i.e., calls that are not in superposition) is bounded polylogarithmically.

Towards that end, we define a relation between procedures to account for recursion. Given a program $\PR\triangleq \D::\ST$, the call relation $\to_{\PR}\ \subseteq \textnormal{Procedures} \times \textnormal{Procedures}$ is defined for any two procedures $\proc_1,$ $\proc_2 \in \ST$ as $\proc_1 \to_{\PR} \proc_2$ whenever $\proc_2 \in \ST^{\proc_1}$. The relation $\succeq_{\PR}$ is then the transitive closure of $\to_{\PR}$. The relation $\sim_{\PR}$ is defined by $\proc_1 \sim_{\PR} \proc_2 $ if $\proc_1\succeq_{\PR} \proc_2$ as well as $\proc_2 \succeq_{\PR} \proc_1$ both hold. Finally, the relation $\succ_{\PR}$ is defined as $\proc_1 \succ_{\PR} \proc_2 $ if $\proc_1\succeq_{\PR} \proc_2$ and  $\proc_1 \not\sim_{\PR} \proc_2$ both hold. A procedure $\proc $ is \emph{recursive} whenever $\proc\sim_{\PR} \proc$ holds. Two procedures $\proc $ and $\proc'$ are \emph{mutually recursive} whenever $\proc\sim_{\PR} \proc'$ holds


\begin{definition}\label{def:half}
A program $\PR$ is said to be \emph{recursively halving}, denoted $\PR\in \half{}$, if for each procedure $\proc \in\PR$ and for each procedure call $\tbt{call } \proc'(\sa_1,\ldots,\sa_n);\in\ST^{\proc}$,
\[
\text{if }\proc\sim_\PR \proc' \text{ then there are } 1 \leq i \leq n \text{ and  } {\sa} \text{ such that }{\sa}^\boxminus \text{ or } {\sa}^\boxplus \text{ appears in }\sa_i.
\]
\end{definition}
This restriction can be viewed as a recursion scheme, which implies a polylogarithmic time bound on programs in $\half{}$ by ensuring that in every (mutually) recursive procedure call at least one of the input qubit lists is cut in half.

Now we impose a further condition on the number of  sequential (mutually) recursive procedure calls. For that purpose, we define the \emph{width} of a program $\PR$ in the following way.

\begin{definition}\label{def:width}
Given a program $\PR$, the \emph{width} of a procedure $\proc \in \PR$ is defined by $\width_\PR(\proc) \triangleq w^\proc_\PR(\ST^\proc)$ where $w^\proc_\PR(\ST)$ is defined inductively on statements by:
\begin{align*}
	w^\proc_\PR(\tbt{skip};) & \triangleq 0 \\
	w^\proc_\PR(\q \asg \U^g(\ia);) &\triangleq 0\\
	w^\proc_\PR(\ST_0 \,\ST_1) & \triangleq w^\proc_\PR(\ST_0)+w^\proc_\PR(\ST_1)\\
	w^\proc_\PR(\cif{\bb}{\{\ST_1\}}{\{\ST_0\}}) & \triangleq \max(w^\proc_\PR(S_0), w^\proc_\PR(\ST_1))\\
	w^\proc_\PR(\qcase{\q}{\ST_0}{\ST_1}) & \triangleq \max(w^\proc_\PR(\ST_0), w^\proc_\PR(\ST_1))\\
	w^\proc_\PR(\call \proc'(\sa_1,\ldots,\sa_n);) & \triangleq \begin{cases}1, \text{ if } \proc \sim_\PR \proc', \\ 0,\text{ otherwise.}\end{cases}
\end{align*}
The $\emph{width}$ of a program $\width(\PR)$ is defined by $\width(\PR) \triangleq \max_{\proc\in\PR} \width_\PR(\proc)$. Let $\wi_{\leq 1}$ be defined by $\wi_{\leq 1} \triangleq \{ \PR \ | \ \width(\PR)\leq 1\}$.
\end{definition}

\begin{definition}
The set $\plp$ of \emph{PolyLog Programs}  is defined by $\plp \triangleq \half{} \cap \wi_{\leq 1}$. 
\end{definition}


As we will see in next Section (Theorem~\ref{thm:soundnesscompleteness}), the restriction to $\plp$ programs ensures that they can be simulated by quantum random-access Turing machines running in polylogarithmic time.

\begin{example}
Both programs \texttt{SEARCH} and \texttt{SQLOG} of Figure~\ref{fig:lfoq-examples} can be shown to be in $\plp$. 
Let us consider the case of \texttt{SQLOG} which has two recursive procedures: \pn{f} and \pn{g}.
We have $\pn{f} \sim_{\texttt{SQLOG}} \pn{f} \succ_{\texttt{SQLOG}} \pn{g} \sim_{\texttt{SQLOG}} \pn{g}$.
It is easy to check that $\texttt{SQLOG}\in\half$ as the procedure \pn{f} recursively calls itself on $\bq_1^\boxplus$ and \pn{g} recursively calls itself on $\bq_1^\boxminus$.
Furthermore, we verify that:
\begin{align*}
\width(\texttt{SQLOG}) &= \max(\width_{\texttt{SQLOG}}(\pn{f}), \width_{\texttt{SQLOG}}(\pn{g}))\\
&=\max(w_{\texttt{SQLOG}}^{\pn{f}}(\ST^\pn{f}), w_{\texttt{SQLOG}}^{\pn{g}}(\ST^\pn{g}))\\
&=\max(0+w_{\texttt{SQLOG}}^{\pn{f}}(\tb{call }\pn{f}(\bq_1^\boxplus , \bq_2);)+w_{\texttt{SQLOG}}^{\pn{f}}(\tb{call }\pn{g}(\bq_1, \bq_2\ominus\el{1});),\\
&\phantom{=\max(} \max(w_{\texttt{SQLOG}}^{\pn{g}}(\tb{call }\pn{g}(\bq_1^\boxminus , \bq_2);), w_{\texttt{SQLOG}}^{\pn{g}}(\bq_2[1]\asg \unot;)))\\
&=\max(0+1+0, \max(1, 0))=1
\end{align*}
\end{example}

\subsection{Properties of PLP Programs}


Because of the $\half$ condition, programs in $\plp$ can be shown to be terminating.

\begin{theorem}\label{thm:termination}
If $\PR\in \plp$, then $\PR$ terminates.
\end{theorem}
\begin{proof}
This is trivially ensured by the {\half} condition, which shows that the size of arguments is strictly decreasing in recursive calls whenever the arguments are not empty as seen in the semantics of $\sa^\boxminus/\sa^\boxplus$ in Figure~\ref{table:semnatbool}, and the program semantics of Figure~\ref{table:operationalsemantics}, as procedure calls terminate on empty list.
\end{proof}

%
As $\plp$ programs are quantum programs, they must be reversible.
We show that we can constructively define a $\plp$ program $\PR^{-1}$ that computes the inverse of $\PR$.

\begin{theorem}[Reversibility]
\label{thm:reversibility}
There exists a program transformation $\cdot^{-1}$ such that, for any $\PR\in\plp{}$, $\sem{\PR^{-1}}\circ \sem{\PR}$ is the identity and $\PR^{-1}\in\plp{}$.
\end{theorem}
\begin{proof}
The program transformation can be constructively defined on program statements. For instance, $(\q \asg \U^g(\ia);)^{-1} \triangleq \q \asg (\U^g(\ia))^{\textnormal{\textdagger}};$ and $(\ST_0\ \ST_1)^{-1}    \triangleq \ST_1^{-1}\ \ST_0^{-1}$.
\end{proof}

Note that Theorems~\ref{thm:termination} and~\ref{thm:reversibility}  can also be obtained as corollaries of the $\fbqpolylog$-soundness that will be proved in Theorem~\ref{thm:soundnesscompleteness}.
We consider the ensured properties (termination and reversibility) as consistency checks before going into the complexity results.

%
%
%
%

\section{A Characterization of Quantum Polylog Time}
\label{s:fbqpolylog}

In this section, we will show that {\plp} characterizes exactly the functions in {\fbqpolylog},  the class of quantum polylog time approximable functions. That is, programs in {\plp} compute functions in {\fbqpolylog} (Soundness) and, reciprocally, for any function in {\fbqpolylog}, there exists a {\plp} program that computes it (Completeness).

\subsection{Quantum Random Access Turing Machines and Polylog Time}
\label{sec:qratm}

To define the class {\fbqpolylog}, we introduce the computational model of \emph{quantum random-access} Turing machines (QRATMs)~\cite{yamakami2022,yamakami2024}. {\fbqpolylog} is not defined on standard quantum Turing machines because, due to its sub-linear time complexity, such a machine would not be able to access all of its input. Random-access machines solve part of the problem by allowing the machine to jump over its input. 


A QRATM has a random access input tape, a log-space index tape, and $c$ work tapes and is then defined as a triplet $(Q,\Sigma,\delta)$, where $Q$ is a finite set of states containing an initial state $s_0$ and two (disjoint) subsets $Q_\text{acc}$ and $Q_\text{rej}$ for accepting and rejecting states, $\Sigma=\{0,1,\#\}$ is the tape alphabet, and the  transition function $\delta$ is such that \[\delta : Q\times \Sigma^{1+c} \to (Q\times \Sigma^{1+c} \times \{L,R,N\}^{1+c} \to \mathbb{C}).\] This transition maps the state and read symbols on the index tape and work tapes to a function mapping each state, each written symbol, and each head move to an amplitude.
Note that the input tape does not have a tape head, hence is not taken into account in this transition function.

To get access to any character of the input, a special transition of the machine is defined:
when the machine is in a special state $s_{\text{query}}$, the cell of input tape corresponding to the binary number written on the index tape is swapped with the cell under the work tape head, and the machine transitions to a state $s_{\text{accept}}$.
Note that, in contrast with~\cite{yamakami2022,yamakami2024}, the input tape is not read-only as we consider a class of functions rather than decision problems, hence having the modified input be part of the output is necessary.
This allows for example to consider the identity function.

A pure configuration of a QRATM is a tuple $(s, w, w_0, w_1, \ldots, w_c, z_0, z_1, \ldots, z_c) \in Q\times{\Sigma^*}^{2+c}\times \mathbb{Z}^{1+c}$ where $s$ is a state, $w$ the word written on the input tape, assumed to begin in cell $0$, $w_0$ is the word on the index tape, $w_1$, ..., $w_c$ the words written on the work tapes, $z_0$ is the position of the index tape head, $z_1$, ..., $z_c$ the tape head positions for the work tapes, all positions are relative to the first character of the word.
The initial configuration for input $x$ is $\gamma(x) \triangleq (s_0, x, \epsilon, \ldots, 0, \ldots)$.
We call a superposition of pure configurations a surface configuration.
Surface configurations can be written as $\sum_r \alpha_r\ket{r}$, with $r$ ranging over pure configurations, $\alpha_r\in\mathbb{C}$ is the amplitude associated with configuration $r$.
QRATMs are also required to satisfy reversibility and well-formedness condition: a configuration may have only one predecessor, and 
the transition function must preserve the norm of configurations, that means that for all reachable surface configurations, $\sum_r |\alpha_r|^2 = 1$.

A QRATM halts in time $t$ on input $x$ if, starting from the initial configuration $\gamma(x)$, after $t$ steps, its surface configuration is a superposition of pure configurations in accepting states.
If for all input $x$, a QRATM $M$ halts on input $x$ in time $T(\size{x})$, we say that $M$ halts in time $T$.
In particular, if there exists $k\in\mathbb{N}$ such that $M$ halts in time $O(\log^k(n))$, $M$ halts in \emph{polylog time}.
If a QRATM halts, its output is defined as the linear combination of the words on the input tape and work tapes, using the previous notations, it corresponds to $\sum_r \alpha_r \ket{w^r, w_0^r, w_1^r, \ldots, w_c^r}$.
Given a function $f:\{0,1\}^*\to\{0,1\}^*$, we say that a QRATM $M$ approximates $f$ with probability $p$ if for all input $x\in\{0,1\}^*$, starting from the initial configuration $\gamma(x)$, $M$ halts with an output $\sum_r \alpha_r \ket{w^r, w_0^r, w_1^r, \ldots, w_c^r}$ such that $\sum_{r\in Q_{acc}\times\{f(x)\}\times \Z^{2+c}} |\alpha_r|^2 \geq p$.

\begin{definition}
The class $\fbqpolylog$ is defined as the set of functions $f:\{0,1\}^*\to\{0,1\}^*$ such that there exists a QRATM  approximating $f$ with probability at least $\frac{2}{3}$ in polylog time.
\end{definition}

Although a natural theoretical model for describing polylogarithmic time, QRATMs are problematic because they are too semantic in nature: to characterise their complexity, the time bound must be given explicit and their halting condition depends on the inner state of the machine, which is a semantic condition. This motivates the provided characterization in the next section.

\subsection{Main Result}


We denote by $\sem{\plp}$ the set of functions computed by programs in $\plp$.
That is $\sem{\plp} \triangleq \{\sem{\PR} \mid \PR\in\plp\}$.
A program $\PR$ approximates function $f:\{0,1\}^*\to\{0,1\}^*$ with probability $p\in [0,1]$ if $\forall x \in \{0,1\}^*, |\braket{f(x)}{\sem{\PR}(x)}|^2 \geq p$, in other words if for all input, the output of $\PR$ coincides with $f$ with probability at least $p$.
The set of functions that can be approximated with probability at least $p$ is denoted by $\sem{\plp}_{\geq p}$.

\begin{theorem}[Soundness \& Completeness]\label{thm:soundnesscompleteness}
$\llbracket \plp{} \rrbracket_{\geq \frac{2}{3}} =  \fbqpolylog$.
\end{theorem}
\begin{proof}
Soundness is proved via the fact that the $\half$ and $\wi_{\geq 1}$ restrictions imply a polylogarithmic bound for the depth of the call tree and a logarithmic bound on the degree of this tree.
Then we show that if some $\plp{}$ program $\PR$ approximates $f$, then the poly-logarithmic time QRATM simulating $\PR$ also approximates $f$ and guarantees that $f\in\fbqpolylog$.
%

Conversely, for completeness, given a polylog time QRATM $M$, we define a \plp{} program that simulates $M$.
To achieve this, we represent the input tape, the index tape and the work tapes using qubit lists, the state of the Turing machine is encoded inside the work tape by writing it to the left of the character under the tape head.
To simulate the execution of $M$, we define the following procedures:
\begin{itemize}
\item \texttt{access\_input}: allows for QRATM-like access to the input tape by performing quantum branching on each cell of the index tape. The correct cell on the input tape to be read is determined by `splitting' in half the set of possible input tape addresses according to the value of each index tape cell.
\item \texttt{local\_step}: simulates a constant-time transition of $M$ locally on three adjacent cells of the index tape and the work tape, or calls \texttt{access\_input} for the query step.
\item \texttt{full\_step}: performs \texttt{local\_step} iteratively to simulate a transition of $M$ over the entirety of the index and work tapes.
\item \texttt{iterate}: executes \texttt{full\_step} a polylogarithmic number of times, simulating the entire run of $M$.
\end{itemize}

These procedures can be combined to obtain a \plp{} program simulating $M$, by encoding $M$'s tapes in a way that allows for a local evolution of the state.
\end{proof}

\section{Circuit Compilation}
\label{s:compilation}

In this section, we sketch an algorithm that compiles {\plp} programs into circuits of polynomial size and polylogarithmic depth.
An implementation of this algorithm, that also works for non-$\plp$ programs, is available at \url{https://gitlab.inria.fr/mmachado/pfoq-compiler}. The two $\plp$ programs \texttt{SEARCH} and \texttt{SQLOG} displayed in Figure~\ref{fig:lfoq-examples} are also provided in the repository. 
Here, we describe the salient points of the compilation and show that the circuit obtained indeed has polylogarithmic depth.

The compilation strategy takes inspiration from~\cite{HPS23,HPS24} which in particular uses ancillas to factorize the circuits representing procedure calls in branches so as to prevent exponential blow-up of the circuit size.
Their technique is called \emph{anchoring and merging} as when a procedure call is first encountered, an ancilla is associated to this call (\emph{anchoring}), and when a subsequent call to this procedure happens with an input of the same size, this second call is then merged with the first (\emph{merging}).
This way, instead of doubling the size of the circuit whenever recursive calls appear in separate branches of a {\tbt{qcase}}, as in programs \texttt{SEARCH} and \texttt{SQLOG} of Figure~\ref{fig:lfoq-examples}, the size grows linearly in the number of nested recursive calls, hence preventing the exponential blow-up in complexity from the use of the quantum control statement~\cite{Yuan2022}.

Figure~\ref{fig:merging} exemplifies this phenomenon on the \texttt{SEARCH} program: the circuit on the left represented by a grey and white box the circuit for the \texttt{search} procedure applied to $\bq_1, \bq_2$. Since this procedure has two calls to itself, its natural compilation gives an in-depth duplication of the calls. The anchoring/merging process entails a single recursive compilation at the price of an overhead in terms of ancillary qubits and permutations.

\subsection{Outline of the Compilation Algorithm}\label{ss:compilsketch}

The compilation algorithm takes as input a program $\PR\in\plp$ together with a list of the sizes $\bar{s}\triangleq [s_1, \ldots, s_n]$ of its inputs.
As in the semantics, the algorithm maintains a function that maps qubit list variables to lists of pointers to their qubits.
Since the values of this function only depend on the size of the qubit variable, the generation of the circuit does not need to take qubit values into account.

The main compilation algorithm works by induction on the structure of the program statement.
Compiling statements such as $\sa[\ia] \asg \U^g(\ja);$ is straightforward: use the semantics to compute the wire number and which quantum gate to put into the circuit. Compiling a sequence is naturally done by composing the circuits obtained from compiling each statement. For compiling an $\tbt{if}$ statement, note that Booleans only depend on constants and the size of qubit lists, and hence can be computed from the knowledge of list $\bar{s}$.

A $\tbt{qcase}$ statement is compiled using controlled operations: consider $\qcase{\sa[\ia]}{\ST_0}{\ST_1}$, the circuit compiled for $\ST_0$ should be controlled negatively on the wire computed for $\sa[\ia]$, the circuit compiled for $\ST_1$ should be controlled positively on the same wire.
To keep track of those controls, a structure is maintained that lists the control qubits and their state.
The compilation of a non-recursive $\tbt{call}$ consists in compiling the statement of the procedure after substituting its arguments by their expressions provided the qubit lists are non empty.

The only case that introduces complexity is that of recursive calls.
A naive compilation strategy for recursive calls with a recursive procedure calling itself in both branches of a $\tbt{qcase}$, which is allowed by the $\wi_{\leq 1}$ restriction, would yield a number of gates in the compiled circuit exponential in the recursion depth.
To prevent this blow-up, the process (anchoring and merging) maintains a dictionary of ancillaries that associates an ancillary qubit to pairs of procedure names and sizes of the inputs: if the dictionary does not have a key for the considered procedure call, an anchoring ancillary is created, this ancillary is initialized through a \emph{multiple controlled-NOT} (Toffoli gate) encoding the \tbt{qcase} control considered and used to control the quantum circuit of the procedure statement; if the key is present in the dictionary, this ancillary is updated with another Toffoli gate and the two procedure calls are automatically merged as illustrated in Figure~\ref{fig:merging}. 
This strategy is sound as the $\wi_{\leq 1}$ restriction ensures that two repeated calls always occur in orthogonal branches and can be simply combined in the same ancillary qubit.

To show the polylogarithmic depth bound on circuits implementing \plp{} programs, 
we first need to demonstrate that merging in the context of {\plp} can be done in polylogarithmic depth.
Second, we show that the $\wi_{\geq 1}$ and $\half$ restrictions imply that the number of nested recursive procedure calls is polylogarithmically bounded.

\subsection{Compilation to a Circuit of Polylog Depth}\label{ss:overview}


\begin{figure}
\centering
\scalebox{1}{
\begin{tikzpicture}
\node[] (A) at (0,0) {
\begin{quantikz}[wire types = {b}]
\lstick{$\phantom{}\bq_1$} & \gate[style={fill=\gatecolor}]{}\vqw{1}& \rstick{\phantom{$\bar{\q}$}}\\
\lstick{$\phantom{}\bq_2$} & \gate{} & \rstick{\phantom{$\bar{\qa}$}}
\end{quantikz}};

\node[anchor=south west] (B) at (3,0.5) {
\scalebox{1}{
\begin{quantikz}[wire types = {q,q}, row sep = 3mm, column sep = 3mm]
\lstick[8]{$\bq_1$}&  & & \gate[3,style={fill=\gatecolor}]{}&\\[-10pt]
&  & & &\\[-10pt]
&  & & &\\[-2mm]
& \octrl{1} & \octrl{1} & \ctrl{-1} &\\[-5pt]
& \octrl{1} & \ctrl{4} &  \octrl{-1} &\\[-5pt]
& \gate[3,style={fill=\gatecolor}]{} \vqw{3}& & &\\[-10pt]
&  & & &\\[-10pt]
&  & & &\\[-2mm]
\lstick{$\bq_2$}& \gate{} &\targ{}&\gate{}\vqw{-4}&
\end{quantikz}}};

\node[anchor=north west] (C) at (3,-0.5) {
\scalebox{1}{
\begin{quantikz}[row sep = 3mm, column sep = 3mm]
\lstick[8]{$\bq_1$}&  & & &\gate[8]{}\permute{6,7,8,4,5,1,2,3} & &\gate[8]{}\permute{6,7,8,4,5,1,2,3} & & &\\[-10pt]
& & & & & & & & &\\[-6pt]
& & & & & & & & &\\[-2mm]
& \octrl{1} & \octrl{1} & \ctrl{1} & & & &\ctrl{1} & \octrl{1} &\\[-6pt]
& \ctrl{4} & \octrl{5} & \octrl{6} & & & &\octrl{6} & \octrl{5} &\\[-9pt]
&  & & && \gate[3,style={fill=\gatecolor}]{}\vqw{3} & & & &\\[-10pt]
& & & & & & & & &\\[-10pt]
& & & & & & & & &\\[-2mm]
\lstick{${\bq_2}$}& \targ{} & & & &\gate{} & & & &\\[-1mm]
\lstick{$\ket{0}$} & & \targ{} & \targ{} & &\ctrl{-1} & & \targ{} & \targ{} & \\[-1mm]
\lstick{$\ket{0}$} & & & \targ{}& \ctrl{-3}& & \ctrl{-3} & \targ{} & & 
\end{quantikz}}};

\draw[->] (A) to [bend left=20]  node[yshift = 5mm,xshift=0mm] {\textbf{(a)} in-depth} (B.west)  ;
\draw[->] (A) to [bend right=20] node[yshift =-5mm,xshift=0mm] {\textbf{(b)} merging} (C.west);
 \end{tikzpicture}}
 \captionsetup{justification=centering}
 \caption{Compilation strategies for \texttt{search} defined in Figure~\ref{fig:lfoq-examples}}
\label{fig:merging}
\end{figure}
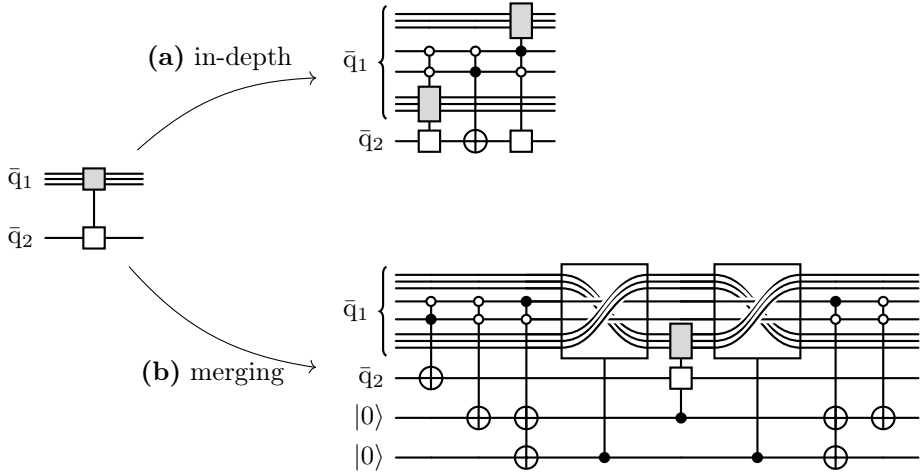


In a {\plp} program, a recursive procedure call always cuts one input qubit list in half.
This means that the procedure will work either on the first half or on the second half of the qubit list.
In the worst case, which is also the most typical, the procedure will be called recursively on each half depending on some condition.
To avoid doubling the treatment of those calls, the anchoring and merging process merges those calls in such a way that the subcircuit for the procedure on half as many qubits is able to work on both halves.
This implies conditionally swapping the two halves.

For instance, consider the \texttt{SEARCH} program in Figure~\ref{fig:lfoq-examples} performing binary search. The procedure body of \texttt{search}, in the recursive case (i.e. where the input $\bq_1$ has size at least 4) consists of three (non-trivial) quantum cases. According to the state of the control qubits, we either (1) perform a recursive call to $\texttt{search}$ on the second half of $\bq_1$, (2) apply a $NOT$ gate to $\bq_2$, or (3) perform a recursive call to $\texttt{search}$ on the first half of $\bq_1$.
Compiling these three branches in sequence incurs a recursive doubling of instances of $\texttt{search}$, as shown on the left circuit of Figure~\ref{fig:merging} where the circuit corresponding to the \texttt{search} procedure is symbolized by a gray box and a white box.
This doubling can be avoided by \emph{merging} the two calls to \texttt{search} in the same circuit, using controlled-swap gates (also called Fredkin gates).
This new circuit, given on the right of Figure~\ref{fig:merging}, contains only a single call to \texttt{search}.
Note that two ancillas are used: one for controlling whether the recursive \texttt{search} is executed, the other for controlling the swapping between the first and second half.
This accounts for a constant added cost for each recursive call.
In addition, the cost of the controlled permutation between the first and second halves of $\bq_1$ should be taken into account.


\begin{lemma}
\label{lem:permutation}A controlled permutation on $n$ qubits can be performed by a quantum circuit of size $O(n)$ and depth $O(\log n)$.
\end{lemma}
\begin{proof}[Proof]
Any permutation can be written as the composition of two sets of disjoint transpositions, and therefore any permutation can be performed in constant time, using two time steps~\cite{MN01}. To perform a \emph{controlled} permutation, it suffices to create $O(n)$ ancillas with the correct controlled state, which can be done in $O(\log n)$ depth with $O(n)$ gates.
\end{proof}

Note that the ancillas used to control the permutation are linear in number but can be reused for nested recursive calls, hence the total number of ancillary qubits used for compiling a program on $n$ qubits will be linear in $n$.



In the case of Figure~\ref{fig:merging}, we are able to merge the two instances of \texttt{search} since they have the same input size and therefore encode precisely the same unitary operation, up to a renaming of qubits. In a general program, one needs to account for the total number of procedures that differ either in procedure name, input size and integer input.
We prove that this number is polylogarithmically bounded.

\begin{lemma}
\label{lem:plpbound}
A procedure call occurring in an \plp{} program on $n$ input qubits results in $O(\log n)$ calls to mutually recursive procedures with unique sizes.
\end{lemma}
In the above lemma, the number of input qubits is obtained by summing the sizes of (i.e., the number of qubits in) each operand of the procedure call.
From those results, we obtain that the compilation process produces a quantum circuit of polynomial size and polylog depth in the size of the inputs.

\begin{theorem}[Compilation]
\label{thm:comp}
Given a \plp{} program $\PR$, and input size $n=\sum_{\bq \in Var(\PR)} |\bq|$, the quantum circuit produced by the compilation process is of size $O(n\, \textnormal{polylog}(n))$ and depth $O(\textnormal{polylog}(n))$.
\end{theorem}

\begin{example}
To illustrate the compilation algorithm and the polylogarithmic depth bound, consider the \texttt{SEARCH} program from Figure~\ref{fig:lfoq-examples}.
The compilation of this program with $\size{\bq_1}=14$ gives the quantum circuit depicted in Figure~\ref{fig:searchqc}.
Each recursive call yields a constant number of controlled gates and makes use of $2$ ancillary qubits: $1$ for anchoring and $1$ to swap the first and second half.
Thus the circuit obtained has depth $O(\log \size{\bq_1})$ and a number $O(\log \size{\bq_1})$ of ancillary qubits.
\end{example}

\begin{figure}
\centering
\newcommand{\qun}{\lstick[14]{$\bq_1$}}
\newcommand{\qde}{\lstick{$\bq_2$}}
\newcommand{\aze}{\lstick{$\ket{0}$}}
\newcommand{\bigperm}{\gate[14][2.2cm]{}\permute{9,10,11,12,13,14,7,8,1,2,3,4,5,6}}
\newcommand{\midperm}{\gate[6][1.2cm]{}\permute{5,6,3,4,1,2}}
\begin{quantikz}[row sep={5 mm,between origins},column sep=1mm]
\qun&         &         &         &\bigperm &          &         &          &\midperm  &\octrl{1}&\midperm  &          &          &\bigperm &         &         & \\[-2mm]
    &         &         &         &         &          &         &          &          &\ctrl{13}&          &          &          &         &         &         & \\
    &         &         &         &         &\octrl{1} &\octrl{1}&\ctrl{1}  &          &         &          &\ctrl{1}  &\octrl{1} &         &         &         & \\[-2mm]
    &         &         &         &         &\octrl{12}&\ctrl{11}&\octrl{12}&          &         &          &\octrl{12}&\octrl{12}&         &         &         & \\
    &         &         &         &         &          &         &          &          &         &          &          &          &         &         &         & \\[-2mm]
    &         &         &         &         &          &         &          &          &         &          &          &          &         &         &         & \\
    &\octrl{1}&\octrl{1}&\ctrl{1} &         &          &         &          &          &         &          &          &          &         &\ctrl{1} &\octrl{1}& \\[-2mm]
    &\octrl{8}&\ctrl{7} &\octrl{9}&         &          &         &          &          &         &          &          &          &         &\octrl{9}&\octrl{8}& \\
    &         &         &         &         &          &         &          &          &         &          &          &          &         &         &         & \\[-2mm]
    &         &         &         &         &          &         &          &          &         &          &          &          &         &         &         & \\
    &         &         &         &         &          &         &          &          &         &          &          &          &         &         &         & \\[-2mm]
    &         &         &         &         &          &         &          &          &         &          &          &          &         &         &         & \\
    &         &         &         &         &          &         &          &          &         &          &          &          &         &         &         & \\[-2mm]
    &         &         &         &         &          &         &          &          &         &          &          &          &         &         &         & \\
\qde&         & \targ{} &         &         &          &\targ{}  &          &          & \targ{} &          &          &          &         &         &         & \\
\aze&\targ{}  &         & \targ{} &         &\ctrl{2}  &\ctrl{-1}&\ctrl{3}  &          &         &          & \ctrl{3} &\ctrl{2}  &         &\targ{}  &\targ{}  & \rstick{$\ket{0}$}\\
\aze&         &         & \targ{} &\ctrl{-3}&          &         &          &          &         &          &          &          &\ctrl{-3}&\targ{}  &         & \rstick{$\ket{0}$}\\
\aze&         &         &         &         &\targ{}   &         &\targ{}   &          &\ctrl{-3}&          &\targ{}   &\targ{}   &         &         &         & \rstick{$\ket{0}$}\\
\aze&         &         &         &         &          &         &\targ{}   &\ctrl{-17}&         &\ctrl{-17}&\targ{}   &          &         &         &         & \rstick{$\ket{0}$} 
\end{quantikz}
\caption{Quantum circuit resulting of compiling the SEARCH program}
\label{fig:searchqc}
\end{figure}
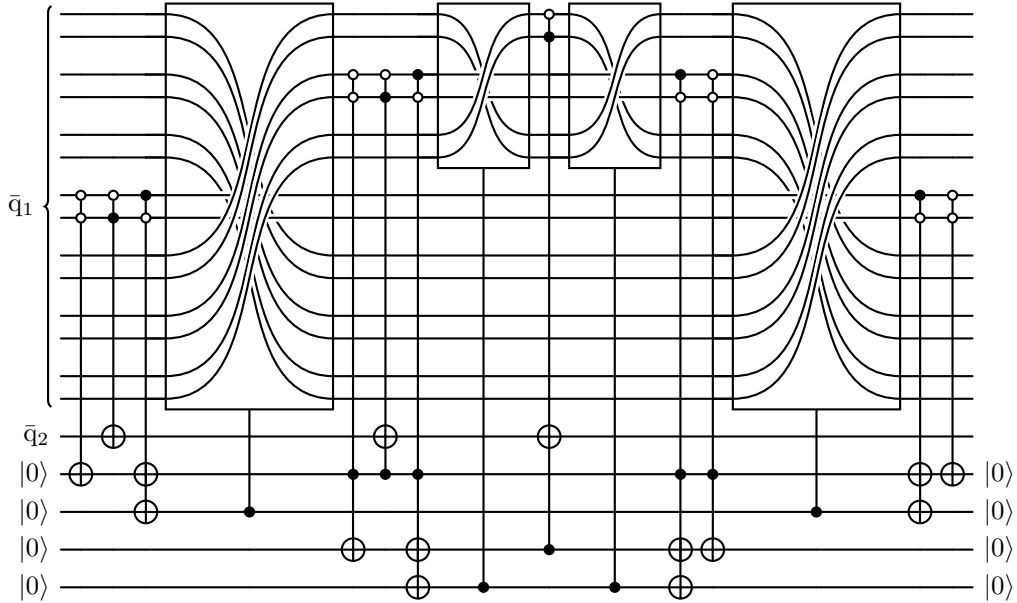
%

\subsection{Limits of Quantum Polylogarithmic Time}\label{ss:strictqnc}

Theorem~\ref{thm:comp} shows that programs in {\plp} can be compiled to quantum circuits of polylog depth and polynomial size, which means that they compute operators in {\qnc}.
In this section, we show that {\plp} is however not complete for {\qnc}, thus recovering the known separation between {\fbqpolylog} and {\qnc}.

{\qnc} is defined in~\cite{MN01} as the union for all $k\in\N$ of the classes of quantum unitary transformations that can be computed by a family of quantum circuits of depth $O(\log^k n)$ with a polynomial number of ancillas.
To compare with {\fbqpolylog}, we define {\bqnc} (for Bounded-error {\qnc}) as in~\cite{CW00} as the class of functions $\{0,1\}^{*}\to\{0,1\}^{*}$ that can be approximated with probability at least $\sfrac23$ by a {\qnc} transformation.

To study the relationship between {\plp} and {\bqnc}, we use the \emph{query model} of quantum computation which is the scenario in which one wishes to compute a function by making use of \emph{black boxes} that are accessed via queries~\cite{A18}. A black box performs a unitary transformation on the quantum state according to some oracle function $\mathcal{O}:\{0,1\}^\ast \to \{0,1\}$, where the operation is usually defined as $\ket{\bar{x},y}\mapsto \ket{\bar{x},y\oplus \mathcal{O}(\bar{x})}$, which is considered to be performed in a single step. 

Given a function $f$, we denote by $Q(f):\N\to\N$ the function that maps $n\in\mathbb{N}$ to the least number of queries necessary for a quantum algorithm to approximate $f$ with bounded-error on inputs of size $n$. The function $Q(f)$ gives a \emph{lower bound} on the time complexity of approximating $f$. For instance, for the OR, AND, and PARITY defined as
\[\arraycolsep=5mm
\begin{array}{ccc}
\text{OR}(\bar{x})     \triangleq \max_{i=1\dots n} x_i &
\text{AND}(\bar{x})    \triangleq \min_{i=1\dots n} x_i &
\text{PARITY}(\bar{x}) \triangleq \bigoplus_{i=1}^n x_i
\end{array}
\]
we have that, while Grover's algorithm~\cite{G96} allows for a quadratic speedup in the query complexity of AND and OR, no speedup exists for PARITY.

\begin{lemma}[\protect{\cite{Z99}}]\label{lem:and-or-query} For $f\in\{\textnormal{AND},\textnormal{OR}\}$, we have that $Q(f)=\Theta(\sqrt{n})$.
\end{lemma}

\begin{lemma}[\protect{\cite{F98,BBCMdW01}}]\label{lem:parity-query}
$Q(\textnormal{PARITY})=\Theta(n)$.
\end{lemma}

In contrast to the lower bounds on the query complexity for $\textnormal{AND}$, $\textnormal{OR}$, and $\textnormal{PARITY}$, we can show that there is an upper bound on the query complexity of functions approximable by programs in $\plp$.
This bound on quantum random-access Turing machines can be obtained by viewing the read-input transition as an oracle query as in~\cite{yamakami2022,yamakami2024} and noting that the access to this oracle is bounded by the step-count of the QRATM.

\begin{lemma}\label{lem:qratmbound}
Let $f\in\sem{\plp}_{\geq\frac23}$. There exists $k\in\mathbb{N}$ such that $Q(f)=O(\log^k n)$.
\end{lemma}

Lemma~\ref{lem:qratmbound} gives a polylogarithmic upper bound on the query complexity of functions in $\sem{\plp}_{\geq\frac23}$.
Lemmas~\ref{lem:and-or-query} and~\ref{lem:parity-query} give lower bounds on the query complexity of $\textnormal{AND}$, $\textnormal{OR}$, and $\textnormal{PARITY}$ that are bigger than polylogarithms, hence proving that those functions are not in $\sem{\plp}_{\geq\frac23}$ even though they can be approximated by circuits of polylogarithmic depth and polynomial size.

\begin{lemma}
\label{lem:notinfbqpolylog}
$\textnormal{AND},\textnormal{ OR},\textnormal{ PARITY}\in\bqnc \setminus \sem{\plp}_{\geq\frac23}$.
\end{lemma}

Combining this result with Theorem~\ref{thm:comp}, we obtain a strict inclusion in $\bqnc$.
Note that this strict inclusion is mandated by the well-known result that $\fbqpolylog\subsetneq\bqnc$.

\begin{theorem}
\label{thm:more-than-polylog}
$\sem{\plp}_{\geq \frac{2}{3}}  \subsetneq \bqnc$.
\end{theorem}

\section{Conclusion and Future Work}

We presented a quantum programming language $\plp$ that captures exactly $\fbqpolylog$, that is functions approximated in polylog time by a QRATM.
This characterization relies on some restrictions, in particular on the arguments of recursive calls to guarantee the complexity bound.
We show a compilation procedure that produces quantum circuits of polylog depth, hence recovering the inclusion in the class of quantum circuits of polylog depth and polynomial size $\qnc$.

\textit{Extending \plp{}}. The conditions imposed in the definition of \plp{} are by no means hard to extend while safeguarding most or even all the results presented here. For instance, the \half{} condition can be relaxed such that we consider halving of the input not at every procedure call but in every closed loop of procedure calls within a given rank.
Such an extension would increase the expressive power of the language, thus allowing a programmer more flexibility.
In this paper we chose to have a streamlined language with restrictions that are simple to check in order to obtain a more readable characterization.

\textit{Characterizing \qnc{}.} To our knowledge, there are currently no implicit characterizations of \bqnc{}. Extending {\plp} to characterize this class would be particularly interesting.
For example, adding a statement for recursively forking on each half of the quantum state would make it possible to capture $\textnormal{AND}$, $\textnormal{OR}$, and $\textnormal{PARITY}$ while still being sound for {\bqnc}.

\bibliography{references}
\end{document}